\documentclass[11pt,reqno]{amsart}
\usepackage{amsthm,amsfonts,amssymb,amsmath,oldgerm,upgreek,xcolor}
\numberwithin{equation}{section}
\usepackage{tikz-cd}
\usepackage{amsmath}
\usepackage{amsfonts}
\usepackage{amssymb}
\usepackage{setspace}
\usepackage{stmaryrd}
\usepackage{mathrsfs}
\usepackage{fancyhdr}
\usepackage{graphicx}
\usepackage[colorlinks = true, linkcolor=blue,
filecolor=darkgreen,
urlcolor=red,
citecolor=red,pdftex]{hyperref}
\usepackage{psfrag}
\usepackage{listings} 
\usepackage{makecell}
\usepackage[top=30mm,bottom=30mm,left=25mm,right=25mm,a4paper]{geometry}
\usepackage{enumitem}
\usepackage{pifont}
\usepackage{enumitem}

\def\eps{\varepsilon }

\newcommand\br{\begin{remark}}
\newcommand\er{\end{remark}}
\newcommand\bp{\begin{pmatrix}}
\newcommand\ep{\end{pmatrix}}
\newcommand{\be}{\begin{equation}}
\newcommand{\ee}{\end{equation}}
\newcommand{\ba}[1]{\begin{array}{#1}}
\newcommand{\ea}{\end{array}}

\newcommand{\beg}{\begin{example}}
\newcommand{\eeg}{\end{exaplem}}
\newcommand{\bpr}{\begin{proposition}}
\newcommand{\epr}{\end{proposition}}
\newcommand{\bt}{\begin{theorem}}
\newcommand{\et}{\end{theorem}}
\newcommand{\bc}{\begin{corollary}}
\newcommand{\ec}{\end{corollary}}
\newcommand{\bl}{\begin{lemma}}
\newcommand{\el}{\end{lemma}}
\newcommand{\bd}{\begin{definition}}
\newcommand{\ed}{\end{definition}}
\newcommand{\brs}{\begin{remarks}}
\newcommand{\ers}{\end{remarks}}

\newcommand{\zstroke}{%
  \text{\ooalign{\hidewidth -\kern-.3em-\hidewidth\cr$z$\cr}}%
}

\newtheorem{theorem}{Theorem}[section]
\newtheorem{proposition}[theorem]{Proposition}
\newtheorem{corollary}[theorem]{Corollary}
\newtheorem{lemma}[theorem]{Lemma}

\theoremstyle{remark}
\newtheorem{remark}[theorem]{Remark}
\theoremstyle{definition}
\newtheorem{definition}[theorem]{Definition}

\newtheorem{example}[theorem]{Example}

\def\eps {\varepsilon}

\newcommand\R{\mathbb R} 
\newcommand\C{\mathbb C}
\newcommand{\N}{\mathbb N}

\newcommand{\grandO}{\cO}

\newcommand\cF{{\mathcal F}}

\newcommand\cM{{\mathcal M}}

\newcommand\cO{{\mathcal O}}

\newcommand\cR{{\mathcal R}}
\newcommand\cS{{\mathcal S}}

\newcommand{\Number}{\mathcal{N}}
\usepackage{bbm}
\newcommand{\id}{\mathbbm{1}}

\newcommand{\du}{\mathop{}\!\mathrm{d} u}
\newcommand{\ds}{\mathop{}\!\mathrm{d} s}

\title{
{\large Report}
}

\title[Expansion of the Many-body Quantum Gibbs State]{Expansion of the Many-body Quantum Gibbs State of the Bose--Hubbard Model on a Finite Graph}

\author{Zied Ammari}
\address{Univ Rennes, [UR1], CNRS, IRMAR - UMR 6625, F-35000 Rennes, France.}
\email{zied.ammari@univ-rennes1.fr}
\author{Shahnaz Farhat}
\address{School of Science, Constructor University Bremen, Campus Ring 1, 28759 Bremen, Germany.}
\email{sfarhat@constructor.university}
\thanks{S.F.\ and S.P.\ acknowledge funding by the Deutsche Forschungsgemeinschaft (DFG, German Research Foundation) – project number 505496137. Z.A.\ acknowledges funding by the French National
Research Agency ANR – project number 22-CE92-0013.
}
\author{S\"oren Petrat}
\address{School of Science, Constructor University Bremen, Campus Ring 1, 28759 Bremen, Germany.}
\email{spetrat@constructor.university}

\begin{document}

\begin{abstract}
We consider the many-body quantum Gibbs state for the Bose--Hubbard model on a finite graph at positive temperature. We scale the interaction with the inverse temperature, corresponding to a mean-field limit where the temperature is of the order of the average particle number. For this model it is known that the many-body Gibbs state converges, as temperature goes to infinity, to the Gibbs measure of a discrete nonlinear Schr\"odinger equation, i.e., a Gibbs measure defined in terms of a one-body theory. In this article we extend these results by proving an expansion to any order of the many-body Gibbs state with inverse temperature as a small parameter. The coefficients in the expansion can be calculated  as vacuum expectation values using a recursive formula, and we 
compute the first two coefficients explicitly.
\end{abstract}

\date{\today}
\maketitle




\section{Introduction and Main Result}

The derivation of nonlinear Gibbs measures from many-body quantum mechanics has recently received a lot of attention in the context of the Bose gas with pair interaction in the continuum. For example, in \cite{FKSS1, LRN, RS} convergence for the partition function and reduced density matrices is proven. In these results the interaction is scaled down with the inverse temperature, analogous to the limit we consider in this article. Note that for continuous systems in two and three dimensions, a renormalization of the nonlinear Gibbs measure is necessary (see \cite{FKSS, FKSV, LM, LTR}). In this article, we avoid technical difficulties associated with continuous systems by considering the simpler Bose--Hubbard model which is defined on a lattice. We furthermore consider only a finite graph instead of an infinite system for technical simplicity. The goal of this article is to prove a much more precise convergence statement, namely an order by order expansion of the Gibbs state. 

One motivation for our work is that such higher order expansions have recently been proven for the continuous Bose gas in the mean-field limit at zero temperature for low-lying eigenvalues and the corresponding (excited and ground) eigenstates \cite{bossmann_petrat_seiringer_2021}, and for the dynamics \cite{bossmann_petrat_soffer_2021,Falconi_Nelson}; see also \cite{bossmann_review,BLMP} for reviews of these results, and \cite{Bossmann_2023, BLMP2} for applications. At zero temperature and in the mean-field limit, the leading order is described by the Hartree equation (a nonlinear Schr\"odinger equation with convolution-type nonlinearity). Note that there are also results concerning the dynamics of discrete nonlinear Schr\"odinger (DNLS) equations and their derivation from many-body quantum theory, e.g., in \cite{PPZ}. The next-to leading order is described by Bogoliubov theory. The expansion is then proven by using perturbation theory around Bogoliubov theory. In our main result in this article, the Hartree equation appears as well in the limiting Gibbs measure. However, in contrast to the zero-temperature case, we do not see Bogoliubov theory clearly emerging in our high temperature limit.  

Another motivation for our work is that recently the Kubo--Martin--Schwinger condition for the thermal equilibrium of quantum and classical systems has attracted some interest \cite{ammari:hal-02104091, AS, DV, Ven} (see \cite{AGGL, GV, HHW} for older results). In particular, one can interpret the Gibbs measure as the KMS equilibrium state for the discrete nonlinear Schr\"odinger equation. More generally this concept extends to Hamiltonian systems governed by PDEs (see \cite{AS} and references therein). One of the remarkable properties is that the quantum and classical KMS conditions can be linked to each other rigorously through the same high temperature limit that we are considering here.  Therefore our expansion of the many-body quantum Gibbs state  provides information  on the KMS condition as well.  We briefly describe here this relation and refer the reader to \cite{ammari:hal-02104091} for more details.  The Bose--Hubbard model defines a dynamical system $(\alpha_t,\omega_\varepsilon)$ given by  a group of automorphisms over the algebra of bounded operators,
\begin{equation}
\alpha_t(A)= e^{i\frac{t}{\varepsilon} H_{\varepsilon}} \,A \,  e^{-i\frac{t}{\varepsilon}H_{\varepsilon}} \,,
\end{equation}
where $H_\varepsilon$ is the Bose--Hubbard Hamiltonian defined in \eqref{micro_Hamiltonian}, and by a quantum  Gibbs state defined by 
\begin{equation}
\omega_\varepsilon (A) =\frac{{\rm Tr}(e^{-\beta H_{\varepsilon}} A)}{{\rm Tr}(e^{-\beta H_{\varepsilon}} )}\,.
\end{equation}
It is known that the Gibbs state $\omega_\varepsilon$ is the unique KMS state at inverse temperature $\varepsilon\beta$ of the Bose--Hubbard system satisfying 
\begin{equation}
\label{KMS_cond}
\omega_\varepsilon(A \;\alpha_{i\varepsilon\beta}(B)) =\omega_\varepsilon(BA)\,,
\end{equation}
where $\alpha_{i\varepsilon\beta}$ is an analytic extension of the dynamics to complex times. A specific choice of observables in \eqref{KMS_cond} leads to
\begin{equation}
\label{KMS_W_cond}
\omega_\varepsilon\left( W_\varepsilon(f) \, \frac{\alpha_{i\varepsilon\beta}(W_\varepsilon(g))-W_\varepsilon(g)}{i\varepsilon}\right)=
\omega_\varepsilon\left( \frac{[W_\varepsilon(g), W_\varepsilon(f)]}{i\varepsilon}\right)\,,
\end{equation}
where $W_\varepsilon(\cdot)$ are the Weyl operators in  \eqref{Weyloperator}. According to \cite{ammari:hal-02104091}, taking the high temperature limit ($\varepsilon \to 0$)  in the relation \eqref{KMS_W_cond} yields  the classical KMS condition 
studied by G.~Gallavotti and E.~Verboven \cite{GV},
\begin{equation}
\label{KMS_class}
\beta\; \int_{\ell^2(G)}
  e^{i \Re e\langle f,u\rangle} \, \{ e^{i \Re e\langle g,u\rangle}, h(u)\} 
 \;d\mu_\beta(u)\,
=\int_{\ell^2(G)}   \; \{ e^{i \Re e\langle g,u\rangle}, e^{i \Re e\langle f,u\rangle}\}
 \;d\mu_\beta(u)\,,
 \end{equation}
where $f,g\in \ell^2(G)$, $h$ is the Hamiltonian of the DNLS equation in \eqref{eq.ham.dnls},  $\mu_\beta$ is the Gibbs measure, and $\{\cdot,\cdot\}$ is the Poisson bracket. Hence our main Theorem \ref{Theorem1}  provides 
an optimal rate of convergence and an expansion of both sides of \eqref{KMS_W_cond} in terms of the inverse temperature parameter.

Another flavor of this topic is its relationship with entropy and  with Berezin quantization  on symplectic manifolds. In fact the Gibbs state and respectively the Gibbs measure are minimizers of their corresponding von Neumann and Boltzmann entropies, so that we can approach the high temperature limit in terms of these variational problems (see \cite{LRN}). On the other hand,  the high temperature limit can be interpreted as a classical limit for Gibbs states in the framework of deformation quantization (see, e.g., \cite{BSP,Ven}). Such a problem was recently studied  in \cite{Ven} and a convergence (to the leading order without expansion) similar to our result  is proven in \cite[Proposition 3.3]{Ven}. The argument is based on Berezin--Lieb and Peierls--Bogolyubov inequalities \cite{Li}.  It would be interesting to apply our method to these two approaches.

\subsection{General framework}
Let $G=(V,E)$ be a finite (undirected and simple) graph, with $V$ the set of vertices and $E$ the set of edges. The degree of a vertex $x\in V$ is denoted by $\deg(x)$. We consider the one-body (complex) Hilbert space ${{\ell}}^2(G)$, endowed with the standard scalar product and norm
\begin{equation}
\langle u,v \rangle=\sum_{x\in V}  \overline{ u(x)} v(x), \qquad \Vert u\Vert = \left(\sum_{x\in V} |u(x)|^2\right)^{1/2} .
\end{equation}
We will sometimes use the orthonormal basis $\{e_x\}_{x\in V }$ of $\ell^2(G)$ defined by
\begin{equation}
e_x(y)=\delta_{x y},\qquad \forall y \in V. 
\end{equation}
The symmetric Fock space $\cF$ is defined as
\begin{equation}
\cF :=\Gamma (\ell^2(G))=  \underset{m\geq0}{\bigoplus}  \ell^2(G)^{ \otimes^m_s} \simeq \underset{m\geq0}{\bigoplus}  \ell_s^2(G^m) ,
\end{equation}
where $\otimes_s$ denotes the symmetric tensor product, and $\ell_s^2(G^m) $ is the symmetric $\ell^2$ space over $G^m$. Let $a^*_x$ and $a_x$ be the usual creation and annihilation operators satisfying the canonical commutation relations
\begin{equation}
[a_x,a^*_y] = \delta_{x y}, \qquad [a_x,a_y] = 0 = [a^*_x,a^*_y],
\end{equation}
and set, for any $u\in \ell^2 (G)$,
\begin{equation}
 a^*(u)=\sum_{x\in V}  u(x) \, a^*_x, \qquad a(u)=\sum_{x\in V} \overline{u(x)} \, a_x.
\end{equation}
The second quantization of an operator $B$ on ${{\ell}}^2(G)$ with matrix elements $B_{xy} = \langle e_x, B e_y \rangle$ is defined as
\begin{equation}
{\rm d} \Gamma (B):= \sum_{x,y\in V} a^*_x B_{xy} a_y,
\end{equation}
and we define the number operator as
\begin{equation}
\Number := {\rm d} \Gamma (\id) = \sum_{x\in V} a^*_x a_x.
\end{equation}
For any small parameter $\eps > 0$, coupling constant $\lambda> 0$, and chemical potential $\kappa < 0$, we define the Bose--Hubbard Hamiltonian on $\cF$ as
\begin{equation}\label{micro_Hamiltonian}
H_\eps:= \frac{\eps}{2}\sum_{x,y \in V,  \,  x\sim y} (a^*_x-a^*_y)(a_x-a_y) -\eps \kappa    \sum_{x\in V} a^*_x a_x  +\eps^2 \frac{\lambda}{2} \sum_{x\in V} a^*_x a^*_x a_x a_x,
\end{equation}
where $x\sim y$ means that $x$ and $y$ are nearest neighbors.
Note that by introducing the discrete Laplacian $\Delta_d$ as
\begin{equation}
(\Delta_d u)(x):= -\deg(x) u(x)+\sum_{y\in V, \, y\sim x} u(y),
\end{equation}
we can rewrite the Bose--Hubbard Hamiltonian as
\begin{equation}
H_\eps= \eps {\rm d}\Gamma (-\Delta_d -\kappa \id)+\eps^2 \frac{\lambda}{2}  \sum_{x\in V} a^*_x a^*_x a_x a_x.
\end{equation}


\subsection{Main result}

The Gibbs state $\omega_\eps$ at inverse temperature $\beta > 0$ is defined as
\begin{equation}\label{Gibbstate}
\omega_\eps(A): = \frac{1}{Z_\eps} {\rm Tr}(e^{-\beta H_\eps} A)
\end{equation}
for any operator $A$ on $\cF$, where $Z_\eps := {\rm Tr}(e^{-\beta H_\eps})$ is the partition function. Note that $Z_\eps < \infty$ since we chose $\kappa < 0$, see, e.g., \cite[Proposition~5.2.27]{Bratelli-Robinson}. We will keep the inverse temperature as a fixed parameter, and instead consider $\varepsilon$ as the small parameter, so the limit $\eps \to 0$ corresponds to a limit where the inverse temperature and the coupling constant each go to zero in the same way. Our goal is to expand $\omega_\eps$ in powers of $\eps$. In order to write down such a series expansion concretely, let us consider a Weyl operator with the right semiclassical structure. For any $f \in \ell^2(G)$ the Weyl operator $W_\eps(f)$ is defined as
\begin{equation} \label{Weyloperator} 
W_\eps(f) := e^{{i \sqrt{\eps}}\Phi(f)}, \quad \text{with} \quad \Phi(f) := \frac{1}{\sqrt{2}} \Big( a(f)+a^*(f) \Big).
\end{equation}
In our main result Theorem~\ref{Theorem1} we prove an expansion for 
\begin{equation}\label{Gibbsstate_Wf}
Z_\eps \;\omega_\eps(W_\eps(f)) = {\rm Tr}(e^{-\beta H_\eps} W_\eps(f)).
\end{equation}
Note that for $f=0$ this implies an expansion for $Z_\eps$, and both results together can be combined into a single expansion for $\omega_\eps(W_\eps(f))$, see Remark~\ref{rem_combined_expansion}. Note, however, that in the limit $\eps \to 0$ both quantities diverge like $c_\eps := (\eps \pi)^{-|V|}$, therefore we write down the expansion for $c_\eps^{-1} {\rm Tr}(e^{-\beta H_\eps} W_\eps(f))$. The expansion could also be written down, e.g., for reduced one-particle density matrices, see Remark~\ref{rem_reduced_densities}.

The leading order in the expansion is a classical Gibbs measure, defined in terms of the Hamiltonian of the discrete nonlinear Schr\"odinger equation
\begin{equation}
\label{eq.ham.dnls}
h(u) := \langle u, (-\Delta_d)u\rangle -\kappa  \Vert u \Vert^2+\frac{\lambda}{2} \sum_{x\in V} |u(x)|^4.
\end{equation}
We introduce the corresponding nonlinear Hartree operator $h^{\mathrm H}$ as
\begin{equation}\label{Hartree_op}
h^{\mathrm H}(u):= -\Delta_d u -\kappa u + \lambda  |u|^2 u.
\end{equation}
The Gibbs measure is defined as 
\begin{equation}
d \mu_\beta(u) = \frac{1}{z_\beta}  \, e^{-\beta h(u)}  \du,
\end{equation}
where $\displaystyle \du=\prod_{x\in V} \du_x$ with $\du_x$ the Lebesgue measure on $\C$, and $z_\beta= \int_{\ell^2(G)} e^{-\beta h(u)} \, \du$. In our setting it is known that
\begin{equation}
\lim_{\eps \rightarrow 0} \omega_\eps(W_\eps(f))=  \int_{\ell^2(G)} e^{\sqrt{2} i \Re e \langle f, u \rangle} \, d\mu_\beta(u),
\end{equation}
see, e.g., \cite{ammari:hal-02104091}. Our main result is the following.

\begin{theorem}[Higher order expansion]\label{Theorem1}
For any $N\in {\N}$,  $f\in \ell^2(G)$, and $ \eps > 0 $ small enough,  we have 
\begin{equation}
 (\eps \pi)^{|V|}  {\rm Tr} (e^{-\beta H_\eps} W(f))=\sum_{j=0}^N \eps^{{j}}  \int_{\ell^2(G)} e^{\sqrt{2} i \Re e \langle f,u \rangle} C_{{j}}(u,f) e^{-\beta h(u)}   \du  + \grandO(\eps^{N+1}),
\end{equation}
where the $\grandO(\eps^{N+1})$ term may depend on $f$ and the parameters $\beta$, $\kappa$, and $\lambda$. The coefficients  $C_j$ are 
$\varepsilon$-independent and  given as $C_0(u,f)\equiv 1$, and, for $j\geq 1$, 
\begin{equation} \label{expCODD}
\begin{aligned}
C_{{j}}(u,f)& =  
\frac{(-1)^j}{j! \,4^j} \| f \|^{2j}  +  \sum_{m=1}^{2j}  \frac{(-\beta )^m}{m!} \sum_{\ell=m}^{\min(4m-2,2j)} \frac{i^{2j-\ell }}{(2j-\ell)!}   \big \langle A^{(m)}_{{\ell}}(u)    \Omega,  \Phi(f)^{2j-\ell}   \Omega \big \rangle,
\end{aligned}
\end{equation}
where  $\Omega$ is the vacuum in $\cF$ and $A^{(m)}_{{\ell }}(u)$  are defined by the recursive formula in \eqref{recursionrelation}. In particular,
\begin{equation}
C_{{1}}(u,f) = - \frac{\Vert f \Vert^2}{4} -i \frac{\beta}{\sqrt{2}} \langle h^{\mathrm H}(u), f \rangle + \frac{\beta^2}{2} \Vert h^{\mathrm H}(u)  \Vert^2
\end{equation}
and $C_{{2}}(u,f)$ is given in Appendix~\ref{appendix_C_1_C_2}.
\end{theorem}

The theorem is proven in Section~\ref{sec_proof_main_result}.

\begin{remark}\label{rem_combined_expansion}
Note that for $f=0$ we have
\begin{equation}
C_{{j}}(u,0)=   \sum_{m=\left\lceil \frac{1}{2} (j+1) \right\rceil}^{2j}  \frac{(-\beta )^m}{m!}   \big \langle A^{(m)}_{2j}(u)   \Omega,    \Omega \big \rangle,
\end{equation}
which are the coefficients in the expansion of $(\eps \pi)^{|V|} Z_\eps$. Combining both expansions
\begin{align}
(\eps \pi)^{|V|} {\rm Tr} (e^{-\beta H_\eps} W(f)) = \sum_{j=0}^N \eps^j \, \widetilde{C}_j(f) + \grandO(\eps^{N+1}), \\ (\eps \pi)^{|V|} {\rm Tr} (e^{-\beta H_\eps}) = \sum_{j=0}^N \eps^j \, \widetilde{C}_j(0) + \grandO(\eps^{N+1}), 
\end{align}
where $\widetilde{C}_0(0) = z_\beta$, leads to an expansion of the Gibbs state. Denoting by $\alpha\in \N^k$ a multi-index with $|\alpha| := \sum_{i=1}^k \alpha_{i}$, we find
\begin{equation}
\begin{aligned}
 \omega_\eps(W_\eps(f)) &= \frac{\widetilde{C}_0(f)}{z_\beta}+\sum_{j=1}^N \eps^j  \left[ \sum_{\ell= 0}^{j-1} \frac{\widetilde{C}_\ell(f)}{z_\beta} \sum_{k=1}^{j-\ell} \sum_{\substack{\alpha \in \N^k \\ |\alpha| = j-\ell}} \prod_{m=1}^k \left( -\frac{\widetilde{C}_{\alpha_m}(0)}{z_\beta} \right) +\frac{\widetilde{C}_j(f)}{z_\beta}  \right] + \grandO(\eps^{N+1})
\\& =  \int_{\ell^2(G)} e^{\sqrt{2} i \Re e \langle f,u \rangle}  \frac{ e^{-\beta h(u)} \du}{z_\beta} 
\\&\quad +\eps  \int_{\ell^2(G)} e^{\sqrt{2} i \Re e \langle f,u \rangle} \Bigg[ - \frac{\Vert f \Vert^2}{4} -i \frac{\beta}{\sqrt{2}} \langle h^{\mathrm H}(u), f \rangle + \frac{\beta^2}{2} \Vert h^{\mathrm H}(u)  \Vert^2 \Bigg]  \frac{ e^{-\beta h(u)}}{z_\beta}   \du
\\ &\quad  - \eps  \Bigg[ \int_{\ell^2(G)}  e^{\sqrt{2} i \Re e \langle f,u \rangle}  \frac{ e^{-\beta h(u)}  }{z_\beta}\du \Bigg] \Bigg[  \int_{\ell^2(G)}    \frac{\beta^2}{2} \Vert h^{\mathrm H}(u) \Vert^2  \frac{ e^{-\beta h(u)} }{z_\beta} \du  \Bigg] +\grandO(\eps^{2}).
\end{aligned}
\end{equation}
\end{remark}

\begin{remark}\label{rem_reduced_densities}
Note that we could also provide an expansion for expectations of other operators than the Weyl operator $W_\eps(f)$, e.g., for the reduced one-particle density matrix
\begin{equation}
{\rm Tr} (e^{-\beta H_\eps} \eps a^*_x a_y).
\end{equation}
For this example the expansion reads
\begin{equation}
\begin{aligned}
(\eps \pi)^{|V|} {\rm Tr} (e^{-\beta H_\eps} \eps a^*_x a_y) &= \int_{\ell^2(G)}\overline{u}_x \,  u_y \, e^{-\beta h(u)} \du
\\&\quad + \eps  \int_{\ell^2(G)}  \  \Bigg[ \frac{\beta^2}{2} \overline{u}_x \, u_y  \Vert h^{\mathrm H}(u)\Vert^2 -\beta u_y \, \langle h^{\mathrm H}(u),\, e_x\rangle  \Bigg] \ e^{-\beta h(u)} \du \\
&\quad + \grandO(\eps^2) .
\end{aligned}
\end{equation}
\end{remark}

\subsection{Notation and summary of proof}
In addition to the Weyl operator
\begin{equation}
W_\eps(f) = e^{{i \sqrt{\frac{\eps}{2}}}( a(f)+a^*(f) )},
\end{equation}
from \eqref{Weyloperator}, let us introduce the rescaled Weyl operator
\begin{equation}
\widetilde{W}_\eps(u) := e^{\frac{1}{\sqrt{\eps}} ( a^*(u) - a(u) )}.
\end{equation}
With the latter we define the coherent state
\begin{equation}\label{coherentstate}
| u_\eps \rangle := \widetilde{W}_\eps(u) |\Omega \rangle = e^{-\frac{\| u \|^2}{2\eps}} e^{\frac{1}{\sqrt{\eps}} a^*(u)} |\Omega \rangle, 
\end{equation}
where $|\Omega\rangle$ is the vacuum in $\cF$, and the equality follows from the Baker--Campbell--Hausdorff formula.
Then, by direct computation on the finite graph $G$, we have the resolution of identity
\begin{equation} \label{resolofident}
c_\eps \,  \int_{\ell^2(G)} |u_\eps \rangle \,    \langle u_\eps |   \du = 1, 
\end{equation}
where  $c_\eps:= (\eps \pi)^{-|V|}$. The idea of the proof is to insert the identity \eqref{resolofident} in the computation of ${\rm Tr}(e^{-\beta H_\eps} W_\eps(f))$, so that it only contains $\widetilde{W}_\eps(u)^* H_\eps \widetilde{W}_\eps(u)$ instead of $H_\eps$ directly. This Weyl-transformed Hamiltonian can be written as
\begin{align}\label{Hamiltonian_expansion}
\widetilde{W}_\eps(u)^* H_\eps \widetilde{W}_\eps(u) = h(u) + \sum_{j=1}^4 \eps^{j/2} A_j(u),
\end{align}
for some $\eps$-independent operators $A_j(u)$, see Section~\ref{expansion_of_Gibbs_state}. It remains to expand the exponential of \eqref{Hamiltonian_expansion}, and $W_\eps(f)$, in powers of $\eps$. This is done in Section~\ref{sec_Taylor_expansion} where we compute the coefficients of the expansion and thus define the remainder terms. In Section~\ref{sec_proof_main_result} we prove Theorem~\ref{Theorem1} by estimating the remainder terms. Finally, we explicitly compute the first two coefficients of the expansion in Appendix~\ref{appendix_C_1_C_2}.

\textbf{Notation.} In the following sections, we use in our estimates constants $C$ that may depend on the parameters of our model, and that can be different from line to line.

\section{Formal Expansion}

\subsection{Resolution of identity}\label{expansion_of_Gibbs_state}
First, we insert the resolution of identity \eqref{resolofident} into the computation of ${\rm Tr}(e^{-\beta H_\eps} W_\eps(f))$.
\begin{lemma}[Integral formula]\label{Decompositionterms} We have 
\begin{equation} \label{integral formula} 
\frac{1}{c_\eps} {\rm Tr}(e^{-\beta H_\eps} W_\eps(f)) =  \int_{\ell^2(G)}    e^{\sqrt{2} i \Re e \langle f,u \rangle}   \, \langle \Omega ,   e^{-\beta A_\eps(u) }   W_\eps(f)  \Omega \rangle \, e^{-\beta h(u)}  \du,
\end{equation}
 where $W_\eps(\cdot)$ is the Weyl operator defined in \eqref{Weyloperator}, and
 \begin{equation} \label{expA}
   A_\eps(u) :=  \sum_{j=1}^4 \eps^{j/2} A_j (u)
 \end{equation}
 with 
 \begin{align}\label{Coefficients_Ham}
 \begin{split}
  &  A_{1}(u):=  a^*(h^{\mathrm H}(u))+a(h^{\mathrm H}(u)),
\\ & A_{2}(u):=  {\rm d}\Gamma (-\Delta_d)
  -\kappa \Number  
  +\frac{\lambda}{2} \sum_{x\in V} (a^*_x a^*_x u_x^2  +4 a^*_x a_x |u_x|^2 +a_x a_x \bar u_x^2),
\\& A_{3}(u):= \frac{\lambda}{2}    \sum_{x\in V}( a^*_x a^*_x a_x u_x +a^*_x a_x a_x \bar u_x ),
\\ &A_{4}(u):= \frac{\lambda}{2}  \sum _{x\in V} a^*_x a^*_x a_x a_x.
\end{split}
 \end{align}

\end{lemma}
\begin{proof}
We use  the resolution of identity  \eqref{resolofident} to expand 
\vskip 2mm
\begin{align}
\frac{1}{c_\eps}   {\rm Tr}(e^{-\beta H_\eps} W_\eps(f)) 
 & =    \int_{\ell^2(G)} \langle u_\eps, \  e^{-\beta H_\eps} \  W_\eps(f) \  u_\eps \rangle  \du
\nonumber\\ & = \int_{\ell^2(G)} \langle \Omega ,   \  \widetilde{W}_\eps(u)^* e^{-\beta H_\eps } \widetilde{W}_\eps(u)  \  \widetilde{W}_\eps(u)^* W_\eps(f) \widetilde{W}_\eps(u)  \  \Omega \rangle \du
\nonumber\\ & =  \int_{\ell^2(G)} \langle \Omega ,   \  e^{-\beta \widetilde{W}_\eps(u)^* H_\eps \widetilde{W}_\eps(u)}  \  \widetilde{W}_\eps(u)^* \  W_\eps(f)  \  \widetilde{W}_\eps(u) \Omega \rangle \du.
\end{align}
 To expand the above expression in powers of $\eps$, 
we use the shifting properties of the Weyl operator, i.e.,
\begin{align}
\widetilde{W}_\eps(u)^* \, a^*_x \, \widetilde{W}_\eps(u) &= a^*_x +\frac{1}{\sqrt{\eps}}\overline{ u(x)},&
\widetilde{W}_\eps(u)^* \, a_x \, \widetilde{W}_\eps(u) &= a_x +\frac{1}{\sqrt{\eps}} u(x),  \label{Com_x} \\
\widetilde{W}_\eps(u)^* \, a^*(f) \, \widetilde{W}_\eps(u) &= a^*(f) +\frac{1}{\sqrt{\eps}} \langle u,f \rangle,& 
\widetilde{W}_\eps(u)^* \, a(f) \, \widetilde{W}_\eps(u) &= a(f) +\frac{1}{\sqrt{\eps}} \langle f, u \rangle.  \label{Com_f}
\end{align}
Using \eqref{Com_x} gives us directly \eqref{Hamiltonian_expansion} with the coefficients \eqref{Coefficients_Ham}.
Furthermore, using \eqref{Com_f}, we get 
\begin{equation}
\widetilde{W}_\eps(u)^* \,  W_\eps(f)  \,  \widetilde{W}_\eps(u) = { e^{ i\sqrt{\frac{\eps}{2}} \  \widetilde{W}_\eps(u)^* \big(a^*(f)+a(f)\big) \widetilde{W}_\eps(u)}} = e^{\sqrt{2} i \Re e \langle f,u \rangle} \,  W_\eps(f),
\end{equation}
which proves \eqref{integral formula}.
\end{proof}

\subsection{Taylor expansion}\label{sec_Taylor_expansion}

\noindent In order to expand the right-hand side of \eqref{integral formula}, let us write down the Taylor expansions with remainders for $e^{-\beta A_\eps(u)}$ and $W_\eps(f)$. 
For given $N\in \N$, we have 
\begin{subequations}\label{Taylor_separate}
\begin{align}
e^{-\beta A_\eps(u)}&= \underbrace{ \sum_{m=0}^{N}\frac{(-\beta A_\eps(u))^m}{m!}}_{=: \cM^{(N)}_A} + \underbrace{(-\beta A_\eps(u))^{N+1} \int_0^1 e^{-s\beta A_\eps(u)} \frac{(1-s)^N}{N!} \ds}_{=: \cR^{(N)}_A}, \label{Taylor1} \\ 
W_\eps(f)&= \underbrace{\sum_{n=0}^{N} \eps^{\frac{n}{2}}\frac{(i  \Phi(f))^{n}}{n!}}_{=: \cM^{(N)}_f} + \underbrace{(i \sqrt{ \eps} \Phi(f))^{N+1} \int_0^1 e^{i  s \sqrt{ \eps} \Phi(f)} \frac{(1-s)^N}{N!} \ds}_{=: \cR^{(N)}_f}. \label{Taylor2}
\end{align}
\end{subequations}
 The above formulas make sense by functional calculus.   
Note that \eqref{Taylor1} is not an expansion in powers of $\sqrt{\eps}$ yet because $A_\eps(u)$ contains different powers of $\sqrt{\eps}$. Furthermore, according to \eqref{integral formula}, we only need to evaluate $(A_\eps(u))^m$ acting on the vacuum $\Omega$. Recall that $\Omega$ is an analytic vector of the field operator $\Phi(f)$ and it is in the domain of any Wick monomial. Ordering $(A_\eps(u))^m|\Omega\rangle$ in powers of $\sqrt{\eps}$ gives, for $m > 0$,
\begin{equation}\label{powersofA}
A_\eps(u)^m |\Omega\rangle = \sum_{\ell=m}^{4m-2} \eps^{\frac{\ell}{2}} A_{\ell}^{(m)}(u) |\Omega\rangle,
\end{equation}
where $A_{\ell}^{(m)}(u)$ is given by
\begin{equation}  \label{recursionrelation} \begin{aligned}
A_{\ell}^{(m)} &= A_{1}^{(1)} A_{\ell-1}^{(m-1)}+ A_{2}^{(1)}A_{\ell-2}^{(m-1)}+A_{3}^{(1)} A_{\ell-3}^{(m-1)}+A_{4}^{(1)}A_{\ell-4}^{(m-1)} \\
&= \sum_{\underset{ \vert \alpha \vert= \ell}{\alpha \in \{1,2,3,4\}^m}} \prod_{k=1}^m A_{{\alpha_k}},
\end{aligned}
\end{equation}
where the $ A_{{\alpha_k}}(u) \equiv A_{{\alpha_k}}^{(1)}(u)$ for $\alpha_k \in \{1,2,3,4\}$ are given in Lemma~\ref{Decompositionterms}. Here, we have used the multi-index notation, i.e., $|\alpha| = \sum_{k=1}^m \alpha_k$. Note that in \eqref{powersofA} the $\ell = 4m-1$ and $\ell = 4m$ terms vanish since both $A_3(u)$ and $A_4(u)$ contain an $a_x$ acting on the vacuum $|\Omega\rangle$. For later convenience we define $A_{\ell}^{(m)}(u)$ to be zero outside the range of indices in \eqref{powersofA}. To summarize, for any $\ell\in \N$,
\begin{equation}\label{definitionA}
A_{\ell}^{(m)}(u) := \left\{\begin{array}{cl} 0 &, \text{for } \ell <  m \text{ or } \ell > 4m-2, \\ \eqref{recursionrelation} &, \text{otherwise.} \end{array}\right.
\end{equation}
We can now state the expansion of the term $\langle  e^{-\beta A_\eps(u)} \Omega ,  W_\eps(f) \Omega \rangle$ from \eqref{integral formula} in powers of $\sqrt{\eps}$.

\begin{lemma}[Formal expansion]\label{lem_Taylor_expansion}
For any $N \in \N$ we have 
\begin{equation} \label{expC1} 
\langle  e^{-\beta A_\eps(u)} \Omega ,  W_\eps(f) \Omega \rangle = \sum_{j=0}^N \eps^{\frac{j}{2}} C_{\frac{j}{2}}(u,f) + R^{(N)}_\eps(u,f), 
\end{equation}
with  $C_0(u,f)\equiv 1$ and where, for $j\geq 1$,
\begin{equation} \label{expCODDterm}
\begin{aligned}
C_{\frac{j}{2}}(u,f) := \frac{i^{j}}{j!} \langle\Omega, \Phi(f)^{j} \Omega \rangle  +  \sum_{m=1}^{j}  \frac{(-\beta )^m}{m!} \sum_{\ell=m}^{\min(4m-2,j)} \frac{i^{j-\ell }}{(j-\ell)!}  \big \langle A^{(m)}_{\ell}(u)   \Omega,  \Phi(f)^{j-\ell}  \Omega \big \rangle,
\end{aligned}
\end{equation}
with $A^{(m)}_{{\ell }}(u)$ as defined in \eqref{definitionA}. The remainder term $R^{(N)}_\eps(u,f)$ is given as
\begin{equation} \label{boundsremainders}
R^{(N)}_\eps(u,f):=\big \langle \Omega ,\big( \cR^{(N)}_{f,A} + \cM^{(N)}_A \cR^{(N)}_f+\cR^{(N)}_A \cM^{(N)}_f +\cR^{(N)}_A  \cR^{(N)}_f \big) \Omega \big \rangle,
\end{equation}
with $ \cM^{(N)}_A$, $\cR^{(N)}_A$, $\cM^{(N)}_f$, and $\cR^{(N)}_f$ from \eqref{Taylor_separate}, and
\begin{equation}
\cR^{(N)}_{f,A} := \sum_{j=N+1}^{5N-2} \eps^{\frac{j}{2}} \sum_{m=1}^{j}  \frac{(-\beta )^m}{m!} \sum_{\ell=\max(m,j-N)}^{\min(4m-2,j)} \frac{i^{j-\ell }}{(j-\ell)!}  \big \langle A^{(m)}_{\ell}(u)   \Omega,  \Phi(f)^{j-\ell}  \Omega \big \rangle.
\end{equation}
\end{lemma}

\begin{proof}
By using the Taylor expansions \eqref{Taylor1} and \eqref{Taylor2} 
we have 
\[\langle  e^{-\beta A_\eps(u)} \Omega ,  W(f) \Omega \rangle =\big \langle \Omega , \cM^{(N)}_A \cM^{(N)}_f \Omega \big \rangle  +\big \langle \Omega ,\big( \cM^{(N)}_A \cR^{(N)}_f+\cR^{(N)}_A \cM^{(N)}_f +\cR^{(N)}_A  \cR^{(N)}_f \big) \Omega \big \rangle. \]
It remains to multiply out the first term in the expression above. Defining $a_m=(-\beta )^m/ m!$ and $b_{n}=i^{n}/n!$ we have
\begin{align}
\big \langle \Omega , \cM^{(N)}_A \cM^{(N)}_f \Omega \big \rangle &= \sum_{m=0}^N \sum_{n=0}^N  \eps^{\frac{n}{2}} a_m b_{n} \langle A_\eps(u)^m  \Omega, \Phi(f)^{n}  \Omega \rangle \nonumber \\
&=\sum_{n=0}^N  \eps^{\frac{n}{2}} b_{n} \langle\Omega, \Phi(f)^{n} \Omega \rangle +\sum_{m=1}^N \sum_{n=0}^N \sum_{p=m}^{4m-2}  \eps^{\frac{p+n}{2}} a_m b_{n} \langle A_{p}^{(m)}(u) \Omega, \Phi(f)^{n}  \Omega \rangle.
\end{align}
It is now convenient to use the convention \eqref{definitionA}, i.e., $A_{\ell}^{(m)}(u) := 0$ for $\ell <  m$ or $\ell > 4m-2$, since then we can rearrange
\begin{align}
&\sum_{m=1}^N \sum_{n=0}^N \sum_{p=m}^{4m-2}  \eps^{\frac{p+n}{2}} a_m b_{n} \langle A_{p}^{(m)}(u) \Omega, \Phi(f)^{n}  \Omega \rangle \nonumber\\
&\qquad= \sum_{p=0}^{\infty} \sum_{n=0}^N \sum_{m=1}^N \eps^{\frac{p+n}{2}} a_m b_{n} \langle A_{p}^{(m)}(u) \Omega, \Phi(f)^{n}  \Omega \rangle \nonumber\\
&\qquad= \sum_{j=0}^{\infty} \eps^{\frac{j}{2}} \sum_{\tilde\ell=0}^N \sum_{m=1}^N  a_m b_{\tilde\ell}\, \langle A_{j-\tilde\ell}^{(m)}(u) \Omega, \Phi(f)^{\tilde\ell}  \Omega \rangle \nonumber\\
&\qquad= \sum_{j=1}^{5N-2} \eps^{\frac{j}{2}} \sum_{m=1}^j \sum_{\tilde\ell=0}^N  a_m b_{\tilde\ell}\, \langle A_{j-\tilde\ell}^{(m)}(u) \Omega, \Phi(f)^{\tilde\ell}  \Omega \rangle \nonumber\\
&\qquad= \sum_{j=1}^{5N-2} \eps^{\frac{j}{2}} \sum_{m=1}^j \,\sum_{\ell=\max(m,j-N)}^{\min(j,4m-2)} a_m b_{j-\ell}\, \langle A_{\ell}^{(m)}(u) \Omega, \Phi(f)^{j-\ell}  \Omega \rangle,
\end{align}
renaming $\tilde\ell = j-\ell$ in the last step (and using the convention \eqref{definitionA} again). In total, this implies the expansion \eqref{expC1}, noting that $\max(m,j-N) = m$ for $j\leq N$ and $m\geq 1$.
\end{proof}


The expansion from Lemma~\ref{lem_Taylor_expansion} can be simplified by noting that the $\eps^{\frac{j}{2}}$ terms vanish for $j$~odd.

\begin{lemma} [Vanishing of odd terms]\label{vanoddterms}
For all $j$ odd, $C_{\frac{j}{2}}(u,f) = 0$.
\end{lemma}

\begin{proof} 
Recall from \eqref{recursionrelation} that
\begin{align}
A_{\ell}^{(m)} = \sum_{\underset{ \vert \alpha \vert= \ell}{\alpha \in \{1,2,3,4\}^m}} \prod_{k=1}^m A_{{\alpha_k}}.
\end{align}
Now recall from the definition \eqref{Coefficients_Ham} that each term in $A_{{\alpha_k}}$ contains exactly $\alpha_k$ creation or annihilation operators (in normal order). Hence, for $|\alpha|=\ell$, the products $\prod_{k=1}^m A_{{\alpha_k}}$ contain an even (resp.\ odd) number of creation/annihilation operators for $\ell$ even (resp.\ odd). Note that the terms in the products $\prod_{k=1}^m A_{{\alpha_k}}$ are not necessarily normal ordered, but normal ordering does not change the parity of the number of creation/annihilation operators. In the same way, $\Phi(f)^{j-\ell}$ contains only terms with an even (resp.\ odd) number of creation/annihilation operators for $0 \leq j-\ell$ even (resp.\ odd). Since vacuum expectations of terms with an odd number of creation/annihilation operators vanish, we have that
\begin{align}
\langle\Omega, \Phi(f)^{j} \Omega\rangle = 0 = \big \langle A^{(m)}_{\ell}(u)   \Omega,  \Phi(f)^{j-\ell}  \Omega \big \rangle \quad \text{ for } j \text{ odd } 
\end{align}
and any $0 \leq \ell \leq j$, $1 \leq m \leq j$. Thus, all terms in the definition \eqref{expCODDterm} of $C_{\frac{j}{2}}(u,f)$ vanish for $j$~odd.
\end{proof}

Lastly, let us compute the first term of $C_{\frac{j}{2}}(u,f)$ explicitly.
\begin{lemma}\label{reductionoffirstterm}
For the first term in \eqref{expCODDterm} we find
\begin{equation}
\frac{i^{2k}}{(2k)!} \langle\Omega, \Phi(f)^{2k} \Omega \rangle = \frac{(-1)^k}{k! \,4^k} \| f \|^{2k}
\end{equation}
for any $k \in \N$.
\end{lemma}

\begin{proof}
 Such identity is well known and the expression can easily be verified inductively. Alternatively, note that by the Baker--Campbell--Hausdorff formula
\begin{equation}
\sum_{k\geq 0} \frac{1}{(2k)!} \langle \Omega , (a(f)+a^*(f))^{2k} \Omega \rangle = \langle \Omega , e^{a(f)+a^*(f)} \Omega \rangle = e^{\frac{\|f\|^2}{2}} = \sum_{k\geq 0} \frac{\Vert f\Vert^{2k} }{2^k \ k!},
\end{equation}
so comparing coefficients leads to
\begin{equation}
\frac{i^{2k}}{(2k)!} \langle\Omega, \Phi(f)^{2k} \Omega \rangle = \frac{(-1)^{k}}{(2k)!} 2^{-k} \frac{(2k)!}{2^k \, k!} \| f \|^{2k} = \frac{(-1)^k}{k! \,4^k} \| f \|^{2k}.
\end{equation}
\end{proof}

\section{Proof of the Main Result}\label{sec_proof_main_result}

\subsection{Preparatory estimates}\label{sec_technical_estimates}
The following lemmas will be used to estimate the remainder terms in the expansion of the right-hand side of \eqref{integral formula}. Recall that in our Hamiltonian \eqref{micro_Hamiltonian} we used the chemical potential $\kappa < 0$, and recall that we defined the coherent state $u_\eps$ in \eqref{coherentstate}.
\begin{lemma} \label{Hamiltonianestimate}
There is $\alpha $ with $0 < \alpha < -\kappa$ such that 
\begin{equation}\label{coherent_state_number_op}
\langle u_\eps,  \ e^{-\beta   H_\eps }  u_\eps  \rangle \leq  e^{-\beta \alpha \Vert u \Vert^2}
\end{equation}
 for all $\eps$ small enough.
\end{lemma}

\begin{proof}
Note that we can write the Hamiltonian \eqref{micro_Hamiltonian} as $H_\eps=\tilde H_\eps-\kappa\eps \Number$, with $\tilde H_\eps \geq0$, $-\kappa>0$ and    $[\tilde H_\eps, \Number]=0$ in the strong sense. Therefore,
\begin{equation}
\langle u_\eps,  e^{-\beta   H_\eps }  u_\eps  \rangle=  \langle u_\eps,  e^{-\beta   \tilde H_\eps } e^{-\beta (-\kappa) \eps \Number}  u_\eps  \rangle \leq   \langle u_\eps,   e^{-\beta  (-\kappa) \eps \Number}  u_\eps  \rangle. 
\end{equation}
The right-hand side can be computed directly from \eqref{coherentstate} and yields
\begin{align}
\langle u_\eps,  e^{-\beta (-\kappa) \eps \Number}  u_\eps  \rangle & =  e^{-\frac{\Vert u\Vert^2}{\eps} }\langle e^{\frac{a^*(u)}{\sqrt{\eps}}}\Omega ,  e^{-\beta (-\kappa) \eps \Number} e^{\frac{a^*(u)}{\sqrt{\eps}}} \Omega \rangle \nonumber\\
& = e^{-\frac{\Vert u\Vert^2}{\eps}} \sum_{k\geq 0}\sum_{\ell\geq 0} \frac{1}{\ell! k!} \frac{1}{\eps^{k/2} \eps^{\ell/2}} \langle{a^*(u)}^\ell \Omega ,  e^{-\beta (-\kappa) \eps \Number}  {a^*(u)}^k \Omega \rangle \nonumber\\
& =  e^{-\frac{\Vert u\Vert^2}{\eps}}  \sum_{\ell \geq 0} \frac{1}{(\ell !)^2} \frac{1}{\eps^{\ell }}  e^{- \eps \beta(-\kappa) \ell }  \underbrace{ \Vert {a^*(u)}^\ell \Omega\Vert^2}_{=\ell !\Vert u\Vert^{2\ell }} \nonumber\\
& =  e^{-\frac{\Vert u\Vert^2}{\eps}}    \sum_{\ell \geq 0} \frac{1}{\ell !} \bigg(  e^{- \eps \beta (-\kappa)}  \frac{\Vert u\Vert^{2}}{\eps} \bigg)^\ell  \nonumber\\
& =e^{-\Vert u\Vert^2 \big(\frac{1-e^{-\beta (-\kappa) \eps}}{\eps}\big)}.
\end{align}
This directly implies \eqref{coherent_state_number_op}. 
\end{proof}
With the exponential decay from Lemma~\ref{Hamiltonianestimate} we will be able to absorb any polynomial growth in $\| u \|$ that comes from the following estimates.
\begin{lemma}\label{remainderestimate}
For all $m, n \in \N_0$, $s\in [0,1]$,  $f \in \ell^2(G)$, and $\eps>0$ small enough,  there exist $C,q>0$  such that for all $u\in\ell^2(G)$,

\begin{align}
\Vert A_\eps(u)^m  e^{is \sqrt{\eps} \Phi(f)} \Phi(f)^n \Omega \Vert &\leq C \eps^{m/2} \, \langle \Vert u \Vert  \rangle^q,
\end{align}
where  $\langle x \rangle := (1+|x|^2)^{1/2}$ refers to the Japanese bracket. In particular, $q=0$ for $m=0$.
\end{lemma}

\begin{proof}
Let $\cS=\Number+2$. Then we have 
\begin{align}\label{A_f_estimate}
& \big\|  A_\eps(u)^m e^{is \sqrt{\eps} \Phi(f)} \Phi(f)^n \Omega \big\| \nonumber\\ 
&\quad=\big\|  A_\eps(u) \cS^{-2} \cS^2 \, A_\eps(u)\cS^{-4} \cS^{4} \, \cdots \, A_\eps(u)  \cS^{-2m}  \nonumber\\ 
& \qquad\qquad  \cS^{2m} \Phi(f) \cS^{-(2m+1/2)} \cS^{2m+1/2} \, \cdots \, \Phi(f) \cS^{-(2m+n/2)} \cS^{2m+n/2} e^{is \sqrt{\eps} \Phi(f)} \Omega \big\| \nonumber\\ 
& \quad\leq \prod_{k=1}^m  \big\| \cS^{2k-2} A_\eps(u) \cS^{-2k} \big\| \ \prod_{\ell =1}^{n} \big\| \cS^{2m+\ell/2 -1/2} \Phi(f) \cS^{-(2m+\ell/2)} \big\|  \  \big\| \cS^{2m+n/2} e^{is \sqrt{\eps} \Phi(f)} \Omega \big\|.
\end{align}
Before estimating each term in the above expression, note that for any function $F:\N \rightarrow \R$ 
and for any  $f\in \ell^2(G)$ we have
\begin{equation}\label{F_number_commute}
F(\Number)  a^*(f)= a^*(f)  F(\Number+1),\qquad F(\Number+1)  a(f)= a(f)  F(\Number).
\end{equation}
Also recall the standard estimates
\begin{align}\label{a_a_star_bound}
\Vert a(f) \phi  \Vert \leq \Vert f\Vert  \, \Vert( \Number+1)^{1/2} \phi \Vert,\qquad
\Vert a^*(f) \phi  \Vert \leq \Vert f\Vert  \, \Vert( \Number+1)^{1/2} \phi \Vert
\end{align}
for any $\phi \in D(\Number^{1/2})$ and $f\in \ell^2(G)$. We start with the first product in \eqref{A_f_estimate}. Recalling the definition of $A_\eps(u)$ in \eqref{expA} and \eqref{Coefficients_Ham} we note that
\begin{subequations}
\begin{align}
&\big\Vert \cS^{2k-2} A_\eps(u)\cS^{-2k} \big \Vert \nonumber \\ 
&\quad \leq  \eps^{1/2} \, \big\Vert \cS^{2k-2} \  \big[ a(h^{\mathrm H}(u))+a^*(h^{\mathrm H}(u)) \big] \cS^{-2k} \big \Vert  \label{S_A_S_term1}
\\ &\quad\quad+ \eps \, \big\Vert \cS^{2k-2} \  \Big[ { d}\Gamma (-\Delta_d)
  -\kappa \Number
  +\frac{\lambda}{2} \sum_{x\in V} (a^*_x a^*_x u_x^2  +4 a^*_x a_x |u_x|^2 +a_x a_x \overline{u}_x^2) \Big] \cS^{-2k} \big \Vert  \label{S_A_S_term2}
  \\ &\quad\quad + \eps^{3/2} \, \big\Vert \cS^{2k-2} \  \Big[ \frac{\lambda}{2}    \sum_{x\in V}( a^*_x a^*_x a_x u_x +a^*_x a_x a_x \overline{u}_x    )  \Big] \cS^{-2k} \big \Vert \label{S_A_S_term3}
   \\ &\quad\quad + \eps^2 \, \big\Vert \cS^{2k-2} \  \Big[  \frac{\lambda}{2}  \sum _{x\in V} a^*_x a^*_x a_x a_x \Big] \cS^{-2k} \big \Vert. \label{S_A_S_term4}
\end{align}
\end{subequations}
Let us deal with each term separately. Using \eqref{F_number_commute} and \eqref{a_a_star_bound} we find
\begin{align} 
\eps^{-1/2}\eqref{S_A_S_term1} &\leq 
\big\Vert    a(h^{\mathrm H}(u)) \  (\Number+1)^{2k-2} \  (\Number+2)^{-2k} \big \Vert  +          \big\Vert a^*(h^{\mathrm H}(u))  \  (\Number+3)^{2k-2}   \  (\Number+2)^{-2k} \big \Vert
\nonumber\\ &\leq  C \Big( \big\Vert    a(h^{\mathrm H}(u)) \  (\Number+1)^{-2} \ \big \Vert  +          \big\Vert a^*(h^{\mathrm H}(u))    (\Number+1)^{-2}\big\Vert \Big)
\nonumber\\ &\leq C \  \Vert h^{H}(u) \Vert \nonumber\\
&\leq  C \big(  \Vert u \Vert + \Vert u \Vert ^3 \big) .
\end{align}
Similarly we have 
\begin{align}
\eps^{-1}\eqref{S_A_S_term2} &\leq 
\big\Vert  d\Gamma(- \Delta_d)  \cS^{-2} \big \Vert    +\kappa   \big\Vert  \Number \cS^{-2}  \big \Vert +  \frac{\lambda}{2} \sum_{x\in V}  |u_x|^2 \, \big\Vert      a^*_x a^*_x  (\Number +4)^{2k-2}    (\Number+2)^{-2k} \big \Vert \nonumber\\ 
& \quad+   \frac{\lambda}{2} \sum_{x\in V} 4 |u_x|^2 \,  \big\Vert     a^*_x a_x     \cS^{-2} \big \Vert +  \frac{\lambda}{2} \sum_{x\in V}  |u_x|^2 \, \big\Vert       a_x a_x  \Number ^{2k-2}   (\Number+2)^{-2k} \big \Vert \nonumber\\ 
&\leq C \big(1+\Vert u\Vert^2 \big),
\end{align}
as well as
\begin{align}
\eps^{-3/2}\eqref{S_A_S_term3} &\leq 
\frac{\lambda}{2}   \sum_{x\in V}   |u_x  | \,  \big\Vert     a^*_x a^*_x a_x (\Number+3) ^{2k-2} (\Number+2)^{-2k} \big \Vert \nonumber\\
&\quad+ \frac{\lambda}{2}    \sum_{x\in V}   |u_x | \,  \big\Vert     a^*_x a_x a_x (\Number+1) ^{2k-2} (\Number+2)^{-2k} \big \Vert \nonumber \\ 
& \leq  C  \big(1+\Vert u\Vert \big),
\end{align}
and
\begin{align}
\eps^{-2}\eqref{S_A_S_term4} &\leq  
\frac{\lambda}{2}   \sum_{x\in V}    \big\Vert     a^*_x a^*_x a_x a_x (N +2) ^{2k-2} (\Number+2)^{-2k} \big \Vert \leq C.
\end{align}
This implies that
\begin{equation}
\big\Vert \cS^{2k-2} A_\eps(u) \cS^{-2k} \big \Vert \leq C \eps^{1/2} \big( 1+ \|u\|´^3 \big).
\end{equation}
Using the same arguments we find
\begin{align}
\big\Vert \cS^{2m+\ell/2-1/2} \Phi(f) \cS^{-2m-\ell/2} \big \Vert & =  \Big\Vert \cS^{2m+\ell/2-1/2}  \frac{1}{\sqrt{2}}\Big[a(f) +a^*(f)\Big] \cS^{-2m-\ell/2} \Big \Vert 
 \leq  C \Vert f\Vert. 
\end{align}
Combining all estimates into \eqref{A_f_estimate} and noting that
\begin{equation}
e^{-is \sqrt{\eps} \Phi(f)} \Number e^{is \sqrt{\eps} \Phi(f)} \leq 2 \Number + 1 + s^2 \eps \|f\|^2,
\end{equation}
we find
\begin{align}
\big\|  A_\eps(u)^m e^{is \sqrt{\eps} \Phi(f)} \Phi(f)^n \Omega \big\| \leq C \left(\prod_{k=1}^m \eps^{1/2} \big( 1+ \|u\|^3 \big) \right) \left(\prod_{\ell=1}^n \|f\|\right) \leq C \eps^{m/2} \langle \Vert u\Vert \rangle^q
\end{align}
for $q=3m/2$, and where the constant $C$ depends on $m$, $n$, $\Vert f \Vert$, $\kappa$, and $\lambda$.
\end{proof}

\subsection{Remainder estimates}\label{sec_main_estimates}

In this subsection we first provide an estimate for the remainder term $R^{(N)}_\eps(u,f)$ from Lemma~\ref{lem_Taylor_expansion}. Afterwards. we show that the estimate is good enough such that the full remainder in the expansion of the Gibbs state, namely
\begin{equation}\label{full_remainder}
\widetilde{R}^{(N)}_\eps(f) = \int_{\ell^2(G)} e^{\sqrt{2} i \Re e \langle f,u \rangle}   \, R^{(N)}_\eps(u,f) \, e^{-\beta h(u)}  \du,
\end{equation}
is bounded and of order $\eps^{\frac{N+1}{2}}$.

\begin{lemma}[Control of the remainder]\label{lem_remainder_control}
There is $\alpha $ with $0 < \alpha < -\kappa$ such that  $R^{(N)}_\eps(u,f)$  from Lemma~\ref{lem_Taylor_expansion} satisfies the bound
\begin{equation} \label{boundsremainder}
\vert  R^{(N)}_\eps(u,f)\vert \leq \eps^{\frac{N+1}{2}} \, C \langle \Vert u \Vert \rangle^q     \bigg( 1+  e^{\beta \big(h(u) -  \alpha \Vert u \Vert^2\big)} \bigg) 
\end{equation}
for all $\eps>0$ small  enough, where $C$ depends on $N$, $\kappa$, $\lambda$, $\beta$ and $f$.
\end{lemma}

\begin{proof}
By Lemma~\ref{lem_Taylor_expansion}, the expression $R^{(N)}_\eps(u,f)$ is composed of four terms which we estimate separately. We assume $\eps \leq 1$ in the following estimates. Using first some elementary bounds and Cauchy--Schwarz, then Lemma~\ref{remainderestimate} for $s=0$, then $\widetilde{W}_\eps(u)^*  H_\eps \widetilde{W}_\eps(u) = h(u) + A_\eps(u)$, and in the end Lemma~\ref{Hamiltonianestimate}, we find
\begin{align}
\Big \vert  \langle  \Omega , \cR^{(N)}_A  \cM^{(N)}_f  \Omega \rangle \Big \vert  &=  \Bigg \vert   \frac{(-\beta )^{N+1}}{N!}  \sum_{n=0}^{N} \eps^{\frac{n}{2}}\frac{i^{n}}{n!}  \int_0^1 {(1-s)^N}   \Big \langle \Omega ,  A_\eps(u)^{N+1} e^{-s\beta A_\eps(u)}   \Phi(f)^{n}  \, \Omega \Big \rangle    \ds  \Bigg \vert 
\nonumber\\ & \leq  \frac{\beta^{N+1}}{N!}   \sum_{n=0}^{N} \frac{1}{n!}    \int_0^1   \Big \vert  \Big \langle \Omega ,   A_\eps(u)^{N+1} e^{-s\beta A_\eps(u)}   \Phi(f)^{n}\,  \Omega \Big \rangle   \Big \vert     \ds 
\nonumber\\ &   \leq   \frac{\beta^{N+1}}{N!}   \sum_{n=0}^{N} \frac{1}{n!}    \int_0^1 \big\Vert e^{-s\beta A_\eps(u)}   \Omega  \big\Vert \,  \big\Vert  A_\eps(u)^{N+1}    \Phi(f)^{n} \,  \Omega \big\Vert     \ds 
\nonumber\\ &   \leq    \eps^{\frac{N+1}{2}} C \langle \Vert u \Vert \rangle^q      \int_0^1 \sqrt{ \langle \Omega,   e^{-2s\beta A_\eps(u)}   \Omega  \rangle }     \ds
\nonumber\\ &  =  \eps^{\frac{N+1}{2}}   C \langle \Vert u \Vert \rangle^q       \int_0^1 e^{s\beta h(u)} \sqrt{ \langle u_\eps,    e^{-2s\beta H_\eps  }   u_\eps  \rangle }      \ds
\nonumber\\ & \leq   \eps^{\frac{N+1}{2}} C  \langle \Vert u \Vert \rangle^q      \int_0^1 e^{s\beta h(u)}  e^{-s \beta \alpha \Vert u \Vert ^2}  \ds.
\end{align}
The other terms can be estimated similarly, using again Lemmas~\ref{Hamiltonianestimate} and \ref{remainderestimate}. We have 
\begin{align}
 \Big \vert \langle  \Omega ,  \cR^{(N)}_A   \cR^{(N)}_f   \Omega \rangle  \Big \vert  &=  \bigg \vert \eps^{\frac{N+1}{2}}    \frac{(-\beta )^{N+1}}{N!} \frac{i^{N+1}}{N!}  \int_{0}^{1}\mathop{}\!\mathrm{d}\tilde s\, {(1-\tilde s)^N} \int_0^1 \ds\,  {(1-s)^N} \nonumber \\
 &\qquad\qquad\qquad\qquad\qquad\qquad\qquad  \Big \langle \Omega ,     A_\eps(u)^{N+1} e^{-s\beta A_\eps(u)}    e^{i \tilde s \sqrt{\eps} \Phi(f)}   \Phi(f)^{N+1} \, \Omega \Big \rangle \bigg \vert 
\nonumber\\ &\leq  \eps^{\frac{N+1}{2}}  \frac{\beta^{N+1}}{N!} \frac{1}{N!} \int_{0}^{1}\mathop{}\!\mathrm{d}\tilde s \int_0^1 \ds \,  \Vert e^{-s\beta A_\eps(u)}   \Omega  \Vert  \Vert   A_\eps(u)^{N+1}    e^{i \tilde s \sqrt{\eps} \Phi(f)}  \Phi(f)^{N+1} \, \Omega \Vert
\nonumber\\ &\leq  \eps^{{N+1}} C \langle \Vert u \Vert \rangle^q      \int_0^1 e^{s\beta h(u)}  e^{-s \beta \alpha \Vert u \Vert ^2}      \ds,
\end{align}
as well as
\begin{align}
 \Big \vert \langle  \Omega , \cM^{(N)}_A  \cR^{(N)}_f  \Omega \rangle  \Big \vert  &=   \bigg \vert  \eps^{\frac{N+1}{2}}   \frac{i^{N+1}}{N!}  \sum_{m=0}^{N} \frac{(-\beta)^{m}}{m!}  \int_0^1 {(1-s)^N}   \Big \langle \Omega ,    A_\eps(u)^{m}  e^{i  s \sqrt{\eps} \Phi(f)}        \Phi(f)^{N+1} \, \Omega \Big \rangle    \ds    \bigg \vert 
\nonumber\\ & \leq  \eps^{\frac{N+1}{2}}  \frac{1}{N!}  \sum_{m=0}^{N} \frac{\beta^{m}}{m!}    \Vert    A_\eps(u)^m \Omega \Vert \,  \Vert    \Phi(f)^{N+1} \, \Omega \Vert     
\nonumber\\ & \leq  \eps^{\frac{N+1}{2}} C  \langle \Vert u \Vert \rangle^q.
\end{align}
Lastly,
\begin{align}
\Big \vert \langle  \Omega , \cR^{(N)}_{f,A} \Omega \rangle  \Big \vert &= \Bigg \vert\sum_{j=N+1}^{5N-2} \eps^{\frac{j}{2}} \sum_{m=1}^{j}  \frac{(-\beta )^m}{m!} \sum_{\ell=\max(m,j-N)}^{\min(4m-2,j)} \frac{i^{j-\ell }}{(j-\ell)!}  \big \langle A^{(m)}_{\ell}(u)   \Omega,  \Phi(f)^{j-\ell} \, \Omega \big \rangle \Bigg \vert \nonumber\\
&\leq \sum_{j=N+1}^{5N-2} \eps^{\frac{j}{2}} \sum_{m=1}^{j}  \frac{\beta^m}{m!} \sum_{\ell=\max(m,j-N)}^{\min(4m-2,j)} \frac{1}{(j-\ell)!}  \big\| A^{(m)}_{\ell}(u)   \Omega \big\| \, \big\| \Phi(f)^{j-\ell} \, \Omega \big\| \nonumber\\
&\leq \eps^{\frac{N+1}{2}} C \langle \Vert u \Vert \rangle^q,
\end{align}
since $\big\| A^{(m)}_{\ell}(u)   \Omega \big\| \leq C \langle \Vert u \Vert \rangle^q$ by repeatedly applying the standard estimates \eqref{a_a_star_bound}. To conclude the bound \eqref{boundsremainder} recall that $\alpha$ from Lemma~\ref{Hamiltonianestimate} satisfies $0 < \alpha < -\kappa$, and that $h(u) \geq -\kappa \|u\|^2$ (with $\kappa<0$), hence $h(u) - \alpha \Vert u \Vert^2 > 0$ and
\begin{equation}
\int_0^1 e^{s \beta \big(h(u) - \alpha \Vert u \Vert^2\big)} \ds \leq e^{\beta \big(h(u) - \alpha \Vert u \Vert^2\big)}.
\end{equation}
\end{proof}

From Lemma~\ref{lem_remainder_control} we know that the full remainder $\widetilde{R}^{(N)}_\eps(f)$ from \eqref{full_remainder} is of order $\eps^{\frac{N+1}{2}}$. It remains to prove the integrability in $u$.
\begin{lemma}[Integrability of the remainder]\label{precision}
For all $ N \in \N$ we have 
\begin{equation}
\big| \widetilde{R}^{(N)}_\eps(f) \big| \leq C \eps^{\frac{N+1}{2}}.
\end{equation}
\end{lemma}

\begin{proof}
Applying Lemma~\ref{lem_remainder_control} leads to
\begin{align}
\bigg| \int_{\ell^2(G)}  e^{\sqrt{2} i \Re e \langle f,u \rangle} R^{(N)}_\eps(u,f)  e^{-\beta h(u)} \du \bigg| &\leq C \eps^{\frac{N+1}{2}} \int_{\ell^2(G)}  \langle \Vert u \Vert \rangle^q \bigg( 1+  e^{\beta h(u) - \beta\alpha \Vert u \Vert^2} \bigg) e^{-\beta h(u)} \du \nonumber\\
& \leq C \eps^{\frac{N+1}{2}} \int_{\ell^2(G)}  \langle \Vert u \Vert \rangle^q \bigg( e^{-\beta h(u)} +  e^{- \beta \alpha \Vert u \Vert ^2} \bigg) \du \nonumber\\
& \leq C \eps^{\frac{N+1}{2}}
\end{align}
since $\beta >0$ and $\alpha > 0$.
\end{proof}

We can now put everything together to prove our main theorem.

\begin{proof}[Proof of Theorem~\ref{Theorem1}]
Combining Lemmas~\ref{integral formula} and \ref{lem_Taylor_expansion} we find, for any $M\in\N$,
\begin{equation}
(\eps \pi)^{|V|}  {\rm Tr} (e^{-\beta H_\eps} W(f)) - \sum_{j=0}^M \eps^{\frac{j}{2}}  \int_{\ell^2(G)} \ e^{\sqrt{2} i \Re e \langle f,u \rangle} C_{\frac{j}{2}}(u,f) e^{-\beta h(u)}  \du  =  \widetilde{R}^{(M)}_\eps(f),
\end{equation}
where the first summand in the definition of $C_{\frac{j}{2}}(u,f)$ from \eqref{expCODDterm} is given by Lemma~\ref{reductionoffirstterm}, and with remainder $\widetilde{R}^{(M)}_\eps(f)$ defined in \eqref{full_remainder}. In Lemma~\ref{precision} we have proven that $|\widetilde{R}^{(M)}_\eps(f)| \leq C \eps^{\frac{M+1}{2}}$. All coefficients $C_{\frac{j}{2}}(u,f)$ with odd $j$ vanish by Lemma~\ref{vanoddterms}. Thus, for $M =: 2N+1$ odd we have $|\widetilde{R}^{(2N+1)}_\eps(f)| \leq C \eps^{N+1}$, and for $M =: 2N+2$ even we have $|\widetilde{R}^{(2N+1)}_\eps(f)| \leq C \eps^{N+\frac{3}{2}} \leq C \eps^{N+1}$ as well.
\end{proof}

\vskip 2mm

\appendix 

\section{Computation of the first two coefficients}\label{appendix_C_1_C_2}
Here, we compute the coefficients $C_1(u,f)$ and $C_2(u,f)$ from Theorem~\ref{Theorem1} explicitly. Recall that
\begin{equation}
\begin{aligned}
C_{{j}}(u,f)& =  
\frac{(-1)^j}{j! \,4^j} \| f \|^{2j}  +  \sum_{m=1}^{2j}  \frac{(-\beta )^m}{m!} \sum_{\ell=m}^{\min(4m-2,2j)} \frac{i^{2j-\ell }}{(2j-\ell)!}   \big \langle A^{(m)}_{{\ell}}(u)    \Omega,  \Phi(f)^{2j-\ell}  \, \Omega \big \rangle.
\end{aligned}
\end{equation}
We then find

\begin{align} 
C_{1}(u,f)&= 
- \frac{1}{4}  \Vert f\Vert^2 -\beta \sum_{\ell=1}^{2} \frac{i^{2-\ell }}{(2-\ell)!}   \big \langle A^{(1)}_{{\ell}}(u)    \Omega,  \Phi(f)^{2-\ell} \,  \Omega \big \rangle +  \frac{\beta^2}{2!}  \big \langle A^{(2)}_{2}(u)     \Omega,      \Omega \big \rangle
\nonumber\\ & = -\frac{1}{4}  \Vert f\Vert^2  -\beta i   \big \langle A^{(1)}_{1}(u)     \Omega,   \Phi(f)    \Omega \big \rangle - \beta  \big \langle A^{(1)}_{2}(u)     \Omega,     \Omega \big \rangle+ \frac{\beta^2}{2!}  \big \langle A^{(2)}_{2}(u)     \Omega,     \Omega \big \rangle,
\end{align}
where, by the definition \eqref{recursionrelation}, $A^{(1)}_{1} = A_{1}$, $A^{(1)}_{2} = A_2$, and $A^{(2)}_{2} = A_{1} A_{1}$. Now, recalling the definition \eqref{Coefficients_Ham} and using that $a(g) \Omega=0$ for any $g\in \ell^2(G)$, we have 
\begin{subequations}
\begin{align}
\big \langle A^{(1)}_{1}(u)     \Omega,   \Phi(f)    \Omega \big \rangle &= \frac{1}{\sqrt{2}}   \langle a^*(h^{\mathrm H}(u))   \Omega,  a^*(f)   \Omega \rangle  = \frac{1}{\sqrt{2}} \langle h^{\mathrm H}(u), f \rangle  , \\
\big \langle A^{(1)}_{2}(u)     \Omega,     \Omega \big \rangle &= 0, \\
\big \langle A^{(2)}_{2}(u)     \Omega,     \Omega \big \rangle &= \big \langle  a^*(h^{\mathrm H}(u))   \Omega,    a^*(h^{\mathrm H}(u))    \Omega \big \rangle = \Vert h^{\mathrm H}(u)\Vert^2.
\end{align}
\end{subequations}
To summarize,
\begin{align} 
C_{{1}}(u,f) = - \frac{\Vert f \Vert^2}{4} -i \frac{\beta}{\sqrt{2}} \langle h^{\mathrm H}(u), f \rangle + \frac{\beta^2}{2} \Vert h^{\mathrm H}(u)  \Vert^2.
\end{align}
The second coefficient is given by
\begin{align}
C_2(u,f) &= \frac{1}{32}  \Vert f \Vert^{4}  +  \sum_{m=1}^{4} \frac{(-\beta )^m}{m!} \sum_{\ell =m }^{\min(4m-2,4)} \frac{i^{4-\ell}}{(4-\ell)!}   \langle A^{(m)}_{{\ell}}(u)   \Omega, \Phi(f)^{4-\ell}  \,\Omega\rangle \nonumber\\
\begin{split}
&= \frac{1}{32}  \Vert f \Vert^{4} 
-  \beta \bigg( -\frac{i}{6} \langle A^{(1)}_{1}(u)   \Omega, \Phi(f)^3  \Omega\rangle - \frac{1}{2} \langle A^{(1)}_{2}(u)   \Omega,  \Phi(f)^2 \,\Omega\rangle \bigg)
\\ 
&\quad+ \frac{ \beta^2}{2} \bigg(  -\frac{1}{2} \langle A^{(2)}_{2}(u)   \Omega, \Phi(f)^2 \, \Omega\rangle + i \langle A^{(2)}_{3}(u)   \Omega, \Phi(f) \Omega\rangle
 +  \langle A^{(2)}_{4}(u)   \Omega,  \Omega\rangle \bigg)
\\
&\quad - \frac{ \beta^3}{6}  \Big(   i \langle A^{(3)}_{3}(u)   \Omega, \Phi(f) \Omega\rangle +  \langle A^{(3)}_{4}(u)   \Omega, \Omega\rangle \Big)
+ \frac{ \beta^4}{24}  \langle A^{(4)}_{4}(u)   \Omega, \Omega\rangle.
\end{split}
\end{align}
Using again the definitions \eqref{Coefficients_Ham} and \eqref{recursionrelation} we find explicitly
\begin{align}
\begin{split}
C_2(u,f) &= \frac{1}{32} \Vert f \Vert^4 \\
&\quad -\beta \bigg( -\frac{i}{4\sqrt{2}} \Vert f \Vert^2 \langle h^{\mathrm H}(u),f \rangle - \frac{\lambda}{8} \langle u^2,f^2 \rangle  \bigg) 
\\ &\quad +{\beta^2} \bigg(-\frac{1}{8} \Vert h^{\mathrm H}(u) \Vert^2  \Vert f\Vert^2 -\frac{1}{4}\langle h^{\mathrm H}(u), f \rangle ^2 + \frac{\lambda i }{2 \sqrt{2}} \langle u^2,  h^{\mathrm H}(u) f \rangle \\
&\quad\quad\qquad + \frac{i}{ 2 \sqrt{2}}  \Big\langle    h^{\mathrm H}(u), -\Delta_d f\Big\rangle  - \frac{\kappa i }{2\sqrt{2}}  \langle h^{\mathrm H}(u), f \rangle +\frac{\lambda i }{\sqrt{2}} \langle u  h^{\mathrm H}(u), uf \rangle  + \frac{\lambda^2}{4} \sum_{x\in V} \vert u_x \vert^4  \bigg) 
\\ &\quad -{\beta^3} \bigg( \frac{i}{2\sqrt{2}}  \langle h^{\mathrm H}(u), f \rangle  \Vert h^{\mathrm H}(u) \Vert ^2 +\frac{\lambda}{6}  \langle u^2,h^{\mathrm H}(u)^2 \rangle +  \frac{1}{6} \langle h^{\mathrm H}(u), -\Delta_d  h^{\mathrm H}(u) \rangle  - \frac{\kappa}{6} \Vert h^{\mathrm H}(u)\Vert^2 
\\ &\quad\quad\qquad +  \frac{\lambda}{3} \Vert u  h^{\mathrm H}(u)\Vert^2+   \frac{\lambda}{6} \langle h^{\mathrm H}(u)^2,u^2\rangle \bigg) \\
&\quad+\frac{\beta^4}{ 8} \Vert h^{\mathrm H}(u) \Vert^4.
\end{split}
\end{align}


\begin{thebibliography}{References}	

\bibitem{ammari:hal-02104091}
Z.~Ammari and A.~Ratsimanetrimanana, 
\newblock {\em High temperature convergence of the KMS boundary conditions: the Bose-Hubbard model on a finite graph.}
\newblock {Commun. Contemp. Math.}, 23(5), 2021.


\bibitem{AS}
Z. Ammari and V. Sohinger,
\newblock {\em Gibbs measures as unique KMS equilibrium states of nonlinear
   Hamiltonian PDEs.}
  \newblock {Rev. Mat. Iberoam.}, 39(1):29--90, 2023.


\bibitem{AGGL}
M. Aizenman, G. Gallavotti, S. Goldstein, and J.L. Lebowitz,
\newblock {\em Stability and equilibrium states of infinite classical systems.}
\newblock {Commun. Math. Phys.}, 48(1):1--14, 1976.

\bibitem{bossmann_review}
L. Boßmann,
\newblock {\em Low-energy spectrum and dynamics of the weakly interacting Bose gas.}
\newblock {J. Math. Phys. (Special issue: XX International Congress on Mathematical Physics)}, 63, 061102, 2022.




\bibitem{bossmann_petrat_seiringer_2021}
L. Boßmann, S. Petrat, and R. Seiringer,
\newblock {\em Asymptotic expansion of low-energy excitations for weakly interacting bosons.}
\newblock {Forum Math. Sigma}, 9, e28, 2021.






\bibitem{bossmann_petrat_soffer_2021}
L. Boßmann, S. Petrat, P. Pickl, and A. Soffer,
\newblock {\em Beyond Bogoliubov dynamics.}
\newblock {Pure appl. anal.}, 3:677-726, 2021.




\bibitem{Bossmann_2023}
L. Boßmann and S. Petrat,
\newblock {\em Weak Edgeworth expansion for the mean-field Bose gas.}
\newblock {Lett. Math. Phys.}, 113(4):77, 2023.






\bibitem{BLMP}
L. Boßmann,  N. Leopold,  D. Mitrouskas, and S. Petrat,
\newblock {\em Asymptotic analysis of the weakly interacting Bose gas: A collection of recent results and applications.}
\newblock {In: A. Bassi, S. Goldstein, R. Tumulka, N. Zanghi (eds), Physics and the Nature of Reality. Fundamendal Theories of Physics, vol. 215. Springer, Cham}, 2024.


\bibitem{BLMP2}
L. Boßmann, N. Leopold, D. Mitrouskas, and S. Petrat,
\newblock {\em A Note on the binding energy for Bosons in the mean-field limit.}
\newblock {J. Stat. Phys.}, 191:48, 2024.


\bibitem{Bratelli-Robinson}
O. Bratteli and D. W. Robinson,
\newblock {\em Operator algebras and quantum statistical mechanics 1. C*- and  W*- algebras, symmetry groups, decomposition of states.}
\newblock{Texts and Monographs in Physics. Springer-Verlag, New York, second edition, 1987.}

\bibitem{BSP}
J.B. Bru and  W.S. Pedra, 
\newblock {\em Classical dynamics from self-consistency equations in quantum mechanics.}
\newblock  {J. Math. Phys.}  63(5):052101, 2022.


\bibitem{DV}
N. Drago and C.J.F. van de Ven,
\newblock  {\em DLR–KMS correspondence on lattice spin systems.}
\newblock  {Lett. Math. Phys.} 113:88, 2023.


\bibitem{Falconi_Nelson}
M. Falconi, N. Leopold, D. Mitrouskas, and S. Petrat,
\newblock {\em Bogoliubov dynamics and higher-order corrections for the regularized Nelson model.}
\newblock {Rev. Math. Phys.}, 35(4):2350006, 2023.

\bibitem{FKSS}
J. Fröhlich, A. Knowles, B. Schlein, and V. Sohinger,
\newblock {\em Gibbs measures of nonlinear Schr\"{o}dinger equations as limits of many-body quantum states in dimensions $d\leq 3$.}
\newblock {Commun. Math. Phys.} 356:883--980, 2017. 

\bibitem{FKSS1}
J. Fr\"{o}hlich, A. Knowles, B. Schlein, and V. Sohinger,
\newblock {\em A microscopic derivation of time-dependent correlation functions
   of the $1D$ cubic nonlinear Schr\"{o}dinger equation.}
   \newblock {Adv. Math.}, 353:67--115, 2019.


\bibitem{FKSV}
J. Fröhlich, A. Knowles, B. Schlein, and V. Sohinger,
\newblock {\em The mean-field limit of quantum Bose gases at positive temperature.}
\newblock {J. Am. Math. Soc.}, 35:1, 2021.


\bibitem{GV}
G. Gallavotti and E. Verboven,
\newblock {\em On the classical {KMS} boundary condition.}
\newblock {Nuovo cimento B}, 28(1):274--286, 1975.



\bibitem{HHW}
R. Haag, N.M. Hugenholtz, and M. Winnink,
\newblock {\em On the equilibrium states in quantum statistical mechanics.}
\newblock {Commun. Math. Phys.}, 5:215--236, 1967.


\bibitem{LM}
M. Lewin,
\newblock {\em Mean-field limits for quantum systems and nonlinear Gibbs
   measures.}
\newblock {Proc. Int. Cong. Math.}, 5:3800–3821, EMS Press, Berlin, 2022.


\bibitem{LRN}
M. Lewin, P.T. Nam, and N. Rougerie,
\newblock {\em Derivation of nonlinear Gibbs measures from many-body quantum mechanics.}
\newblock {J. Éc. Polytech. Math.} 2:65--115, 2015.



\bibitem{LTR}
M. Lewin,  P.T. Nam, and N. Rougerie,
\newblock {\em Classical field theory limit of many-body quantum Gibbs states in 2D and 3D.}
\newblock {Invent. Math.}, 224:1--130,  2021.

\bibitem{Li}
E.H. Lieb,
\newblock {\em The classical limit of quantum spin systems.}
\newblock {Commun. Math. Phys.}, 31:327--340, 1973.

\bibitem{PPZ}
E. Picari, A. Ponno, and L. Zanelli,
\newblock {\em Mean field derivation of DNLS from the Bose-Hubbard model.}
\newblock {Ann. Henri Poincar\'{e}}, 23:1525--1553, 2022.


\bibitem{RS}
A. Rout and V. Sohinger,
\newblock {\em A microscopic derivation of Gibbs measures for the 1D focusing cubic nonlinear Schr\"{o}dinger equation.} 
\newblock {Commun. Part. Diff. Eq.}, 48:1008--1055, 2023.


\bibitem{Ven}
C.J.F.  van de Ven,
\newblock {\em Gibbs states and their classical limit.}
\newblock {Rev. Math. Phys.}, in press, 2024. 

\end{thebibliography}
\end{document}